\documentclass{ifacconf}

\usepackage{natbib}        

\usepackage{amsmath,amssymb,amsfonts}

\usepackage{graphicx}
\usepackage{xcolor}

\usepackage{url}
\newtheorem{remark}{Remark}
\newtheorem{theorem}{Theorem}
\newtheorem{corollary}{Corollary}
\newenvironment{proof}{\par\noindent\textbf{Proof.}~}{\hfill$\square$\par}

\begin{document}
\begin{frontmatter}

\title{Natural Gradient Descent for Control} 

\thanks[footnoteinfo]{This work is supported in part by the National Science Foundation under award ECCS-2227311. \\
Farnaz Adib Yaghmaie is supported by the Excellence Center at Linköping–Lund in Information Technology (ELLIIT), ZENITH, and partially by Sensor informatics and Decision-making for the Digital Transformation (SEDDIT). This work was partly performed within the Competence Center SEDDIT-Sensor Informatics and Decision making for the Digital Transformation, supported by Sweden’s Innovation Agency within the research and innovation program Advanced digitalization.}

\author[First]{Ramin Esmzad} 
\author[Second]{Farnaz Adib Yaghmaie}
\author[First]{Hamidreza Modares} 

\address[First]{Department of Mechanical Engineering, Michigan State University, East Lansing, MI 48824, USA (e-mails: \{esmzadra, modaresh\}@msu.edu).}
\address[Second]{Department of Electrical Engineering, Linköping University, 58183 Linköping, Sweden (e-mail: farnaz.adib.yaghmaie@liu.se)}

\begin{abstract}                
 This paper bridges optimization and control, and presents a novel closed-loop control framework based on natural gradient descent, offering a trajectory-oriented alternative to traditional cost-function tuning. By leveraging the Fisher Information Matrix, we formulate a preconditioned gradient descent update that explicitly shapes system trajectories. We show that, in sharp contrast to traditional controllers, our approach provides flexibility to shape the system's low-level behavior. To this end, the proposed method parameterizes closed-loop dynamics in terms of stationary covariance and an unknown cost function, providing a geometric interpretation of control adjustments. We establish theoretical stability conditions. The simulation results on a rotary inverted pendulum benchmark highlight the advantages of natural gradient descent in trajectory shaping. 
\end{abstract}


\end{frontmatter}

\section{Introduction}

Optimization techniques, particularly gradient descent (GD) and its numerous variants~\citep{ruder2017overviewgd,laborde2020lyapunov,rattray1998natural,martens2020new}, have become fundamental in modern control and machine learning. These methods are broadly classified into two categories when applied to control systems: 
\textit{GD-based control}, where gradient methods optimize controller parameters, and 
\textit{controlled GD}, where control-theoretic tools improve the convergence properties of a gradient-based optimizer~\citep{lessard2016analysis,padmanabhan2024analysis,nayyer2022passivity}. 
GD-based methods have demonstrated significant success in learning controllers for uncertain environments~\citep{80202,sutton2018reinforcement,8169685,7752775}, as well as in system identification~\citep{ljung1998system, hardt2018gradient} and adaptive control~\citep{gaudio2019connections,ioannou2012robust,10599619,landau2011adaptive}. 

Despite their effectiveness, a key challenge persists: the closed-loop \textit{trajectory behavior}—how system states evolve over time—is typically an \textit{indirect} outcome of weight tuning or cost-function optimization. Traditional methods rely on defining an objective function, applying an optimization algorithm, and then observing the resulting state trajectories~\citep{Cothren2021DataenabledGF}. If these trajectories do not meet performance expectations, the objective function often requires manual adjustments, leading to a cumbersome iterative tuning process. Furthermore, this process is susceptible to ``reward hacking,'' where the optimized policy achieves a low-cost function value but deviates from the designer’s intended behavior~\citep{skalse2022defining}.

In our previous work~\citep{esmzad2024gdl}, we proposed a fundamentally different perspective by \textit{directly shaping} system trajectories through a \textit{gradient-descent-like closed-loop} approach. We introduced a novel parameterization of the stable closed-loop dynamics
\begin{equation}
    A + B K = I - 2 \Gamma P,
\end{equation}
where \( A ,\: B \) are the system matrices, \( K \) is a feedback gain, and \( \Gamma, P \succ 0 \) represent a GD step and a quadratic Lyapunov cost matrix, respectively. This formulation ensures that the closed-loop system behaves analogously to a GD update applied to the Lyapunov function \( V(x_k) = x_k^\top P x_k \), leading to the explicit trajectory update dynamic 
\begin{equation}
    x_{k+1} = x_k - \Gamma \nabla_{x_k} V(x_k).
\end{equation}
Rather than iteratively adjusting an unknown cost function to produce desired state trajectories, we directly impose a controlled gradient flow on the system states.

This trajectory-oriented perspective offers several advantages. Firstly, it provides explicit trajectory control. The gradient formulation provides direct control over the trajectory shape, eliminating the need for trial-and-error cost function tuning.
Secondly, it quantifies robustness by the step size. For a small step size, closed-loop eigenvalues have small variations. 
Lastly, it unifies with a Linear Quadratic Regulator (LQR). We established theoretical connections between our approach and classical optimal control, showing that LQR solutions can be naturally represented in this setting~\citep{esmzad2024gdl}.

One of the key challenges in the proposed approach is the definition of the preconditioning matrix \(\Gamma\), which is generally non-trivial. To address this issue, we use a \textit{natural GD}~\citep{6790500} framework as the update rule for the state-space evolution. We demonstrate that, in this formulation, \(\Gamma\) corresponds to the covariance matrix of the state vector of the closed-loop system, providing a meaningful interpretation: the magnitude of updates in state space is directly influenced by the uncertainty level, as quantified by the covariance matrix. This insight establishes a direct link between control adjustments and the available information in the system. Furthermore, we derive theoretical stability conditions and step-size constraints to ensure robust implementation. The effectiveness of the proposed approach is validated through simulations on the Quanser rotary inverted pendulum benchmark, demonstrating predictable closed-loop behavior compared to the LQR method.


\textit{Notation:} Let $\mathbb{R}^{m \times n }$ denote the real linear space for all real matrices with dimensions $m \times n$. $\mathbb{N}$ is the set of natural numbers. A positive (semi) definite matrix $P \in \mathbb{R}^{n \times n}$ is denoted by $P \succ 0,\: (P \succcurlyeq 0)$. The transpose of a matrix $Q$ is denoted by $Q^\top$. The eigenvalues of a square matrix $A \in \mathbb{R}^{n \times n}$ are denoted by $\lambda_i(A),\:i=1,...,n$ and the spectral radius of $A$ is denoted by $\rho(A)=\max\{\vert \lambda_{1}\vert,...,\vert \lambda_n \vert\}$. The gradient of the function $f(x) \in \mathbb{R}$ with respect to $x \in \mathbb{R}^{n }$ is denoted by $\nabla_{x}f(x) \in \mathbb{R}^{n }$.

\section{System Model and Problem Description}\label{sec:preliminaries}
\noindent Consider a linear time-invariant (LTI) stochastic discrete-time system of the form 
\begin{align} 
x_{k+1}=A \,x_k+B \, u_k+w_k,
\label{eq:syst}
\end{align} 
where $k \in \mathbb{N}$, $x_k \in \mathbb{R}^n$ is the system's state, and $u_k \in \mathbb{R}^m$ is the control input. In addition, $A$ and $B$ are transition and input matrices of appropriate dimensions, respectively. The noise of the system $\omega_k \in \mathbb{R}^n$ is governed by a Gaussian distribution $\mathcal{N}(0, W)$ where $W \in \mathbb{R}^{n\times n} \succ 0$ is its covariance matrix. Throughout this paper, a linear control policy is used
\begin{align} 
u_{k}=K \,x_k, \quad K \in \mathbb{R}^{m \times n}.
\label{eq:control}
\end{align} 
The dynamical system \eqref{eq:syst} evolves stochastically due to the additive Gaussian noise \( \omega_k \)~\citep{kalman1960new}. Using the linear control policy in \eqref{eq:control}, the distribution of the state $x_k$ at time $k$ is given by \( x_k \sim \mathcal{N}(\mu_k, \Sigma_k) \),
where the mean $\mu_k$ and the covariance $\Sigma_k$ are updated respectively as follows
\begin{align}
\mu_{k+1} &= (A + BK)\mu_k, \\
\Sigma_{k+1} &= (A + BK) \Sigma_k (A + BK)^\top + W.
\end{align}
The term \( (A + BK)\Sigma_k(A + BK)^\top \) captures how the closed-loop dynamics influence the propagation of uncertainty. 

A key challenge in closed-loop controller design using LQR is ensuring that the resulting trajectories exhibit desired convergence properties while minimizing deviations due to disturbances. Traditional methods, such as optimal control and reinforcement learning, typically optimize a predefined cost function to tune \( K \), indirectly shaping the closed-loop behavior. If the resulting optimal control solution is not satisfactory, the cost parameter must be adjusted. However, the map from the cost parameters to the low-level system behavior is unknown, which makes it challenging to achieve the designer's intention by adjusting the cost parameter. In this work, we propose an alternative approach where the closed-loop dynamics are parameterized in a GD-inspired framework. This formulation ensures that the state updates follow a gradient-descent-like behavior. We gain direct control over trajectory shaping rather than relying on cost-function tuning. This approach provides a geometric interpretation of closed-loop control. 

\subsection{Basics of Gradients}
Gradients play a fundamental role in optimization and control. In the context of machine learning and control systems, gradients are used to iteratively adjust parameters in order to minimize a given cost function. Traditional gradient-based methods rely on first-order derivatives to provide a direction of improvement, while more advanced techniques, such as natural gradients,
refine this approach by incorporating information about the underlying geometry of the parameter space.


\subsection{GD Algorithm}
In the standard GD, the update rule for minimizing a cost function \( J(\theta) \) with respect to parameters \( \theta \) is
\begin{align}
    \theta_{k+1} = \theta_k - \alpha \nabla_{\theta_k} J(\theta_k),
\end{align}
where \( \alpha > 0 \) is the step size. This method works in an Euclidean parameter space, assuming that all directions have equal importance and are scaled uniformly. However, when the parameter space is curved or highly anisotropic, standard GD can suffer from slow convergence~\citep{ruder2017overviewgd}.

\subsection{Natural GD Algorithm}
Natural gradients are an extension of GD that accounts for the geometry of the parameter space~\citep{6790500,ollivier2019ekf}. Instead of using the Euclidean gradient, the natural gradient uses an alternative direction based on the local curvature of the space~\citep{martens2020new,pascanu2014revisitingnaturalgradientdeep}. The natural gradient update is defined as
\begin{align}
\theta_{k+1} = \theta_k - \alpha G(\theta_k)^{-1} \nabla_{\theta_k} J(\theta_k),
\end{align}
where \( G(\theta_k) \) is the Fisher Information Matrix (FIM), defined as
    \begin{align}
    G(\theta_k) = \mathbb{E}_{z_k \sim p_{\theta_k}} \left[ \nabla_{\theta_k} \log p_{\theta_k}(z_k) \nabla_{\theta_k} \log p_{\theta_k}(z_k)^\top \right],
    \label{eq:FIM}
    \end{align}
with \( p_{\theta_k}(z_k) \) being the Probability Distribution Function (PDF) parameterized by \( \theta_k \). \( G(\theta_k)^{-1} \) serves as a preconditioner, transforming the gradient into a coordinate system that respects the curvature of the cost surface. 
The preconditioning approach significantly improves convergence in ill-conditioned landscapes, ensuring more stable and efficient parameter updates.



\section{From Basic GD to Natural GD in state-space}\label{sec:methodology}
This section explores the evolution of GD methods applied in state space, beginning from the standard formulation and advancing towards natural gradient techniques that account for system curvature and structure.\vspace{-10pt}
\subsection{Standard GD}
Noise-free linear dynamical systems are considered in \cite{esmzad2024gdl}. GD updates for this case is given by
\begin{equation}
x_{k+1} = x_k - \alpha \nabla_{x_k} V(x_k)
\end{equation}
where $V(x_k)$ and $\alpha > 0$ are the cost function and a step size respectively. This method assumes an isotropic state space where all directions are equally important. However, when the landscape of $V(x_k)$ is highly anisotropic (e.g., elongated level sets), standard GD can suffer from slow convergence and oscillatory behavior.\vspace{-5pt}
\subsection{Preconditioned GD}
Again for noise free linear dynamical system in \cite{esmzad2024gdl}, one can consider a preconditioning matrix $\Gamma$, which rescales the gradient update
\begin{equation}
x_{k+1} = x_k - \Gamma \nabla_{x_k} V(x_k)
\end{equation}
where $\Gamma$ is a positive definite matrix that adjusts the step size in different directions. If $\,\Gamma$ is chosen appropriately, this method can significantly accelerate convergence by better aligning the update with the geometry of $V(x_k)$. However, it requires fine-tuning. \vspace{-5pt}
\subsection{Natural GD}
In this paper, we consider linear systems subject to Gaussian noise \eqref{eq:syst} so we will give the natural GD update for $\mu_k$. The natural gradient method refines preconditioned GD by using the inverse FIM, $G(\mu_k)$, as the preconditioner
\begin{equation}
\mu_{k+1} = \mu_k - \alpha G(\mu_k)^{-1} \nabla_{\mu_k} \mathbb{E}_{x_k \sim \mathcal{N}(\mu_k, \Sigma_k)} \left[V(x_k) \right] \label{eq:ngd_ss}
\end{equation}
where $G(\mu_k)$ accounts for the curvature of the state space and $\alpha >0$. Unlike standard GD, which uses the Euclidean gradient, the natural gradient follows the Riemannian structure of the space, leading to more efficient updates. $G(\mu_k)$ can be replaced by a Hessian approximation or another metric tensor that captures the local curvature of $V(x_k)$. In control applications, such matrices arise naturally in system identification and adaptive control, where the uncertainty in the state of the system can be quantified through inverse covariance matrices as it is outlined in the next section.
\section{Natural GD for Control}\label{sec:theory}
Natural gradient methods have been widely recognized for their ability to enhance convergence rates in optimization by accounting for the curvature of the cost landscape \citep{6790500,martens2020new}. Similarly, in the context of state-space control, the use of a positive-definite preconditioning matrix \( \Gamma \) provides a powerful mechanism for improving convergence rates. The gradient preconditioning adapts the descent direction and magnitude to the underlying geometry of the system dynamics, ensuring more efficient convergence.
\subsection{Computing the FIM}
The FIM, \( G(\mu_k) \), quantifies the curvature of the likelihood \( p(x_k) \) in the state space and provides insight into the sensitivity of the likelihood to state changes. For the stochastic system dynamical system in \eqref{eq:syst}, the PDF of the state $x_k \sim \mathcal{N}(\mu_k, \Sigma_k)$ is given by the following  Gaussian distribution
{\small
\begin{align}
p(x_k)  = \frac{1}{(2\pi)^{n/2} |\Sigma_k|^{1/2}} \exp\left(-\frac{1}{2}(x_k - \mu_k)^\top \Sigma_k^{-1} (x_k - \mu_k)\right).
\end{align}
}
The log-likelihood of $p(x_k)$ is given by
\begin{align}
\log p(x_k) &= -\frac{1}{2}  (x_k - \mu_k)^\top \Sigma_k^{-1} (x_k - \mu_k) \nonumber \\
&-\frac{1}{2} \left( \log |\Sigma_k| + n \log (2\pi) \right).
\end{align}
Taking the gradient of \( \log p(x_k) \) with respect to \( \mu_k \)
\begin{align}
\nabla_{\mu_k} \log p(x_k) = \Sigma_k^{-1} (x_k - \mu_k). \label{eq:nabla_log_px}
\end{align}
Substituting \eqref{eq:nabla_log_px} into the FIM definition in \eqref{eq:FIM} yields
\begin{align}
G(\mu_k) =& \mathbb{E}_{x_k \sim \mathcal{N}(\mu_k, \Sigma_k)} \left[ \Sigma_k^{-1} (x_k - \mu_k) (x_k - \mu_k)^\top \Sigma_k^{-1} \right]\notag \\
=&\Sigma_k^{-1} \mathbb{E}_{x_k \sim \mathcal{N}(\mu_k, \Sigma_k)} \left[ (x_k - \mu_k) (x_k - \mu_k)^\top  \right]\Sigma_k^{-1}\notag\\
=&\Sigma_k^{-1} \Sigma_k \Sigma_k^{-1}=\Sigma_k^{-1}.\label{eq:G_cov}
\end{align}
The FIM, \( G(\mu_k) \), provides a measure of the amount of information that the state \( x_k \) carries about the underlying system trajectories. Here, it reduces to the inverse covariance matrix \( \Sigma_k^{-1} \), reflecting the precision (inverse uncertainty) in the Gaussian model. We will use the covariance matrix of the state of the closed-loop system as a natural preconditioner in the GD-based control framework and in the next section, we derive explicit formulations for designing the control gain \( K \).

\subsection{Controller Design Using Stationary Covariance}
To simplify the derivations and facilitate the analysis of the proposed control framework, we consider the steady state covariance matrix, i.e., \(\Sigma_{k+1} = \Sigma_k = \Sigma\).  Theorem~\ref{th:theorem1} encapsulates the core contribution of this work. \vspace{6pt}
\begin{theorem} \label{th:theorem1}
Consider the system \eqref{eq:syst} and assume that $(A, B)$ is controllable. Assume that $\alpha>0$ is given and $0<\lambda<1$. Let $Y,\: F,\: \Sigma, \: M$ denote a feasible solution to the following linear problem

\begin{subequations}
\label{eq:theorem1}
\begin{align}
    AY + BF &= Y - 2\alpha \Sigma, \label{eq:abky}\\
    \begin{bmatrix}
        \lambda Y & (Y-2\alpha \Sigma)^\top \\ * & Y
    \end{bmatrix} &\succeq 0,\label{eq:lmi_stability}\\
    \begin{bmatrix}
    \Sigma - W & Y-2\alpha \Sigma \\ * & M
    \end{bmatrix} &\succeq 0, \label{eq:lmi_cov} \\
    \begin{bmatrix}
        M & Y\\ * & \Sigma
    \end{bmatrix} &\succeq 0 \label{eq:lmi_M}.
\end{align}
\end{subequations}
Then, the preconditioned natural GD control in \eqref{eq:ngd_ss} with $V(x_k)=\mathbb{E}_{x_k \sim \mathcal{N}(\mu_k, \Sigma_k)}\left[x_k^{\top}Y^{-1} x_k\right]$ and $G(\mu_k)=\Sigma^{-1}$ makes the closed-loop dynamics $A+BK$ with $K=F Y^{-1}$ $\lambda-$contractive (and thus stable).
\end{theorem}

\begin{proof}
Let $P=Y^{-1}$. Consider the following Lyapunov candidate 
\begin{align}
    V(x_k) = \mathbb{E}_{x_k \sim \mathcal{N}(\mu_k, \Sigma_k)}\left[x_k^\top P x_k\right], \quad P\succ 0.
\end{align}
and let $G(\mu_k)=\Sigma^{-1}$. As a result, the natural GD reads
\begin{equation}
   \mu_{k+1} = \mu_k - 2 \alpha \Sigma P \mu_k.
\end{equation}
By selecting $\alpha$, one can ensure the natural GD is contractive and thus stable; i.e. $\rho(I-2\alpha \Sigma P)<1$.

To design the controller gain $K$ from the preconditioned GD control design \eqref{eq:ngd_ss}, we introduce a new parametrization of the closed-loop dynamics as $A+BK=I-\alpha \Sigma P$, which is equivalent to \eqref{eq:abky}. Now, we show that $K=F Y^{-1}$ makes the mean of the closed-loop system $\lambda-$contractive.
The controller gain $K$ makes the closed loop system $\lambda$-contractive, if 
\begin{align*}
    (A+BK)^{\top}P(A+BK) \preceq \lambda P.
\end{align*}
Multiplying both sides of the above inequality with $P^{-1}=Y$ from left and right and using $A+BK=I-2\alpha \Sigma P$, one gets
\begin{align}
    (Y-2\alpha \Sigma)^\top P (Y-2\alpha \Sigma) &\preceq \lambda Y
\end{align}
which is equivalent to \eqref{eq:lmi_stability} using the Schur complement lemma. 
In the preconditioned GD control, we select the FIM as 
$G(\mu_k)=\Sigma^{-1}$. The stationary closed-loop covariance using the controller gain $K$ reads
\begin{equation}
    \Sigma = (A + BK) \Sigma (A + BK)^\top + W \label{eq:closed_cov}
\end{equation}
which couples the unknown variables $\Sigma$ and $K$ and thus should be incorporated in the design. One can write \eqref{eq:closed_cov} as
\begin{align}
    \Sigma = &(A + BK)P^{-1}P \Sigma P P^{-1} (A + BK)^\top + W, \notag\\
    =& (AY+BF)P \Sigma P  (AY + BF)^\top + W. \label{eq:closed_cov_2}
\end{align}
Let 
\begin{align}
    M \succeq Y \Sigma^{-1} Y
    \label{eq:M_lmi}
\end{align}
or equivalently $M^{-1} \preceq P \Sigma P$ meaning that there exists a $\Pi \succeq 0$ such that  $ P \Sigma P-M^{-1}-\Pi=0$. Using  $P \Sigma P=M^{-1}+\Pi$ in \eqref{eq:closed_cov_2}, one gets
\begin{align*}
    \Sigma = & (AY+BF)M^{-1}  (AY + BF)^\top + W\\
    &+(AY+BF)\Pi (AY + BF)^\top 
\end{align*}
where the last term in the right hand side is positive because $\Pi \succeq 0$. As a result, one has
\begin{align}
    \Sigma  &\succeq (Y-2\alpha \Sigma) M^{-1} (Y-2 \alpha\Sigma)^\top+W.
    \label{eq:closed_cov3}
\end{align}
Inequalities \eqref{eq:M_lmi} and \eqref{eq:closed_cov3} are equivalent to \eqref{eq:lmi_M} and \eqref{eq:lmi_cov} respectively by the Schur complement. This completes the proof.
\end{proof}

\begin{theorem}[Range of $\alpha$ for closed-Loop stability]
Consider the dynamical system in \eqref{eq:syst}. Assume that $(A,B)$ is controllable. Assume that the controller gain $K$ is designed according to Theorem \ref{th:theorem1}. Then, the closed-loop dynamics $A+BK$ 
is asymptotically stable if and only if
\begin{equation}
0 < \alpha < \frac{1}{\rho(\Sigma P)}. \label{eq:landamSP}
\end{equation}
\end{theorem}

\begin{proof}
The eigenvalues of $A + BK$ are given by
\begin{equation}
\lambda_i(A + BK) = \lambda_i(I - 2\alpha \Sigma P)= 1 - 2\alpha \lambda_i(\Sigma P), \quad \forall i.
\end{equation}
Stability requires
\begin{equation}
|1 - 2\alpha \lambda_i(\Sigma P)| < 1, \quad \forall i. \label{eq:alpha_stable_condition}
\end{equation}
Note that the eigenvalues of $\Sigma P$; i.e. $\lambda_i(\Sigma P)$ are all real, and that is because $\Sigma P$ is similar to a positive definite matrix. To see this point, take the similarity transform as $T=P^{-\frac{1}{2}}$. Then $T^{-1}\Sigma P T= P^{\frac{1}{2}} \Sigma P^{\frac{1}{2}}$ which is a positive definite matrix and has real eigenvalues. Since $\Sigma P$ is similar to $P^{\frac{1}{2}} \Sigma P^{\frac{1}{2}}$,  $\lambda_i(\Sigma P)$ are also real. As a result, \eqref{eq:alpha_stable_condition} is simplified to
\begin{equation*}
-1 <1 - 2\alpha \lambda_i(\Sigma P) < 1, \rightarrow 0<\alpha <\frac{1}{\lambda_i(\Sigma P)},\quad \forall i 
\end{equation*} 
which simplifies to \eqref{eq:landamSP} 
\end{proof}
\begin{remark}
Note that the range for $\alpha$ in \eqref{eq:landamSP} cannot be verified before design as the values of $\Sigma$ and $P$ depend on the value of $\alpha$. However, this range tells us that smaller values of $\alpha $ will decrease the convergence rate and larger values of $\alpha$ will increase the convergence speed. So, $\alpha$ can serve as an intuitive closed-loop behavior-tweaking variable. 
\end{remark}
\begin{corollary}
\textit{(Sensitivity to Step Size).}  
The sensitivity of the closed-loop eigenvalues to the step size $\alpha$ is given by
\begin{equation}
\frac{\partial \lambda_i}{\partial \alpha} = -2\lambda_{\Sigma P,i}. \label{eq:alpha_sens}
\end{equation}
The sensitivity of eigenvalues increases linearly with $\alpha$. 
\end{corollary}
\begin{proof}
The eigenvalues of $A + BK$ are
\begin{equation}
\lambda_i = 1 - 2\alpha \lambda_{\Sigma P,i}.
\end{equation}
Differentiating with respect to $\alpha$ results in \eqref{eq:alpha_sens}.
\end{proof}

\subsection{Discussion}
Stability conditions ensure the existence of a Lyapunov function that decreases along system trajectories, guaranteeing stability. However, these conditions do not necessarily provide an explicit mechanism for \emph{shaping} the trajectories themselves. In contrast, our framework leverages a GD-like formulation to directly influence how the state evolves at each step, thereby offering more transparency over closed-loop behavior. 

Classical LQR designs the gain matrix by minimizing a quadratic cost function, but there is no straightforward, systematic way to target specific \emph{transient} behaviors beyond fine-tuning the weighting matrices \(Q\) and \(R\). The proposed natural GD approach, on the other hand, utilizes a step-size parameter \(\alpha\) and a preconditioning (covariance) matrix \(\Sigma\) to shape the trajectory explicitly. Adjusting \(\alpha\) translates directly into changing the rate of convergence ($\Sigma$ defines how strongly each direction in the state space is scaled in the gradient step)
where \textbf{smaller \(\alpha\)} yields smoother but slower convergence, 
 and \textbf{larger \(\alpha\)} leads to faster convergence, albeit at the risk of oscillations if too large.

\begin{figure}[!b]
\centerline{\includegraphics[width=0.3\columnwidth]{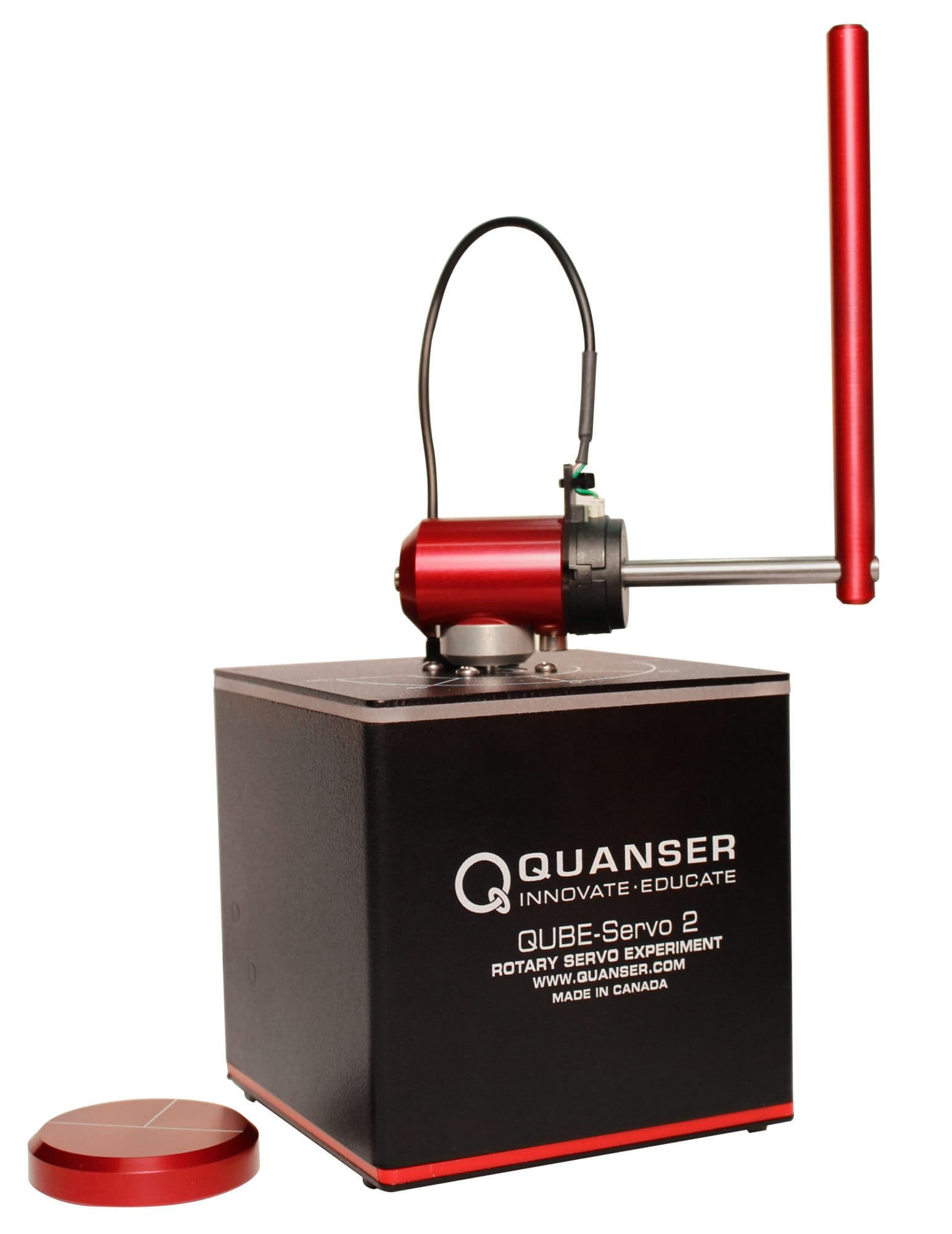}}
\caption{Quanser's Qube-Servo 2 platform}
\label{fig:quanser}
\end{figure}
 
In contrast to methods where the cost function is purely chosen by the designer, our approach \emph{derives} a cost function \emph{by design} from the natural GD-based parameterization. This decouples the choice of cost function from the low-level behavior.


\section{Simulation Results}\label{sec:simulations}
To evaluate the proposed natural gradient control method's performance, we conducted simulations on a widely used benchmark platform \textit{Quanser rotary pendulum (model QUBE-Servo 2)}
\footnote{\url{https://www.quanser.com/products/qube-servo-2/}} 
shown in Figure \ref{fig:quanser}. The system dynamics are represented in discrete time with a sampling interval of \( T_s = 0.01 \) seconds. The state vector is defined as $x = \begin{bmatrix} \theta & \delta & \dot{\theta}   & \dot{\delta} \end{bmatrix}^T$,
where \( \theta \) (rad) is the rotary arm angle and \( \delta \) (rad) is the pendulum angle.

\begin{figure}[!t]
\centerline{\includegraphics[width=1\columnwidth]{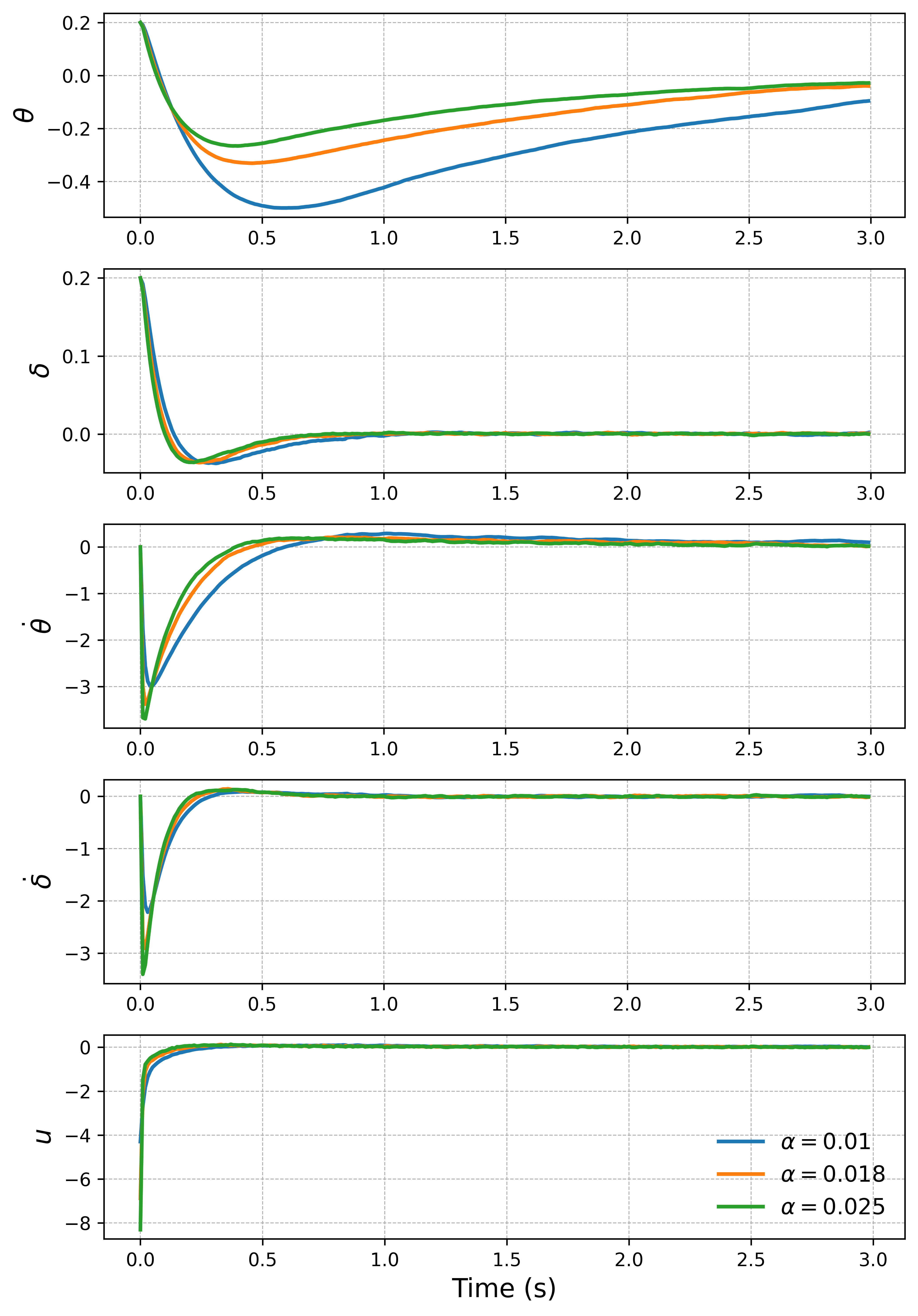}}
\caption{The evolution of the state variables (top four plots) and the control input (the fifth plot) using the natural GD design for various $\alpha$.}
\label{fig:ngd}
\end{figure}

\begin{figure}[!t]
\centerline{\includegraphics[width=1\columnwidth]{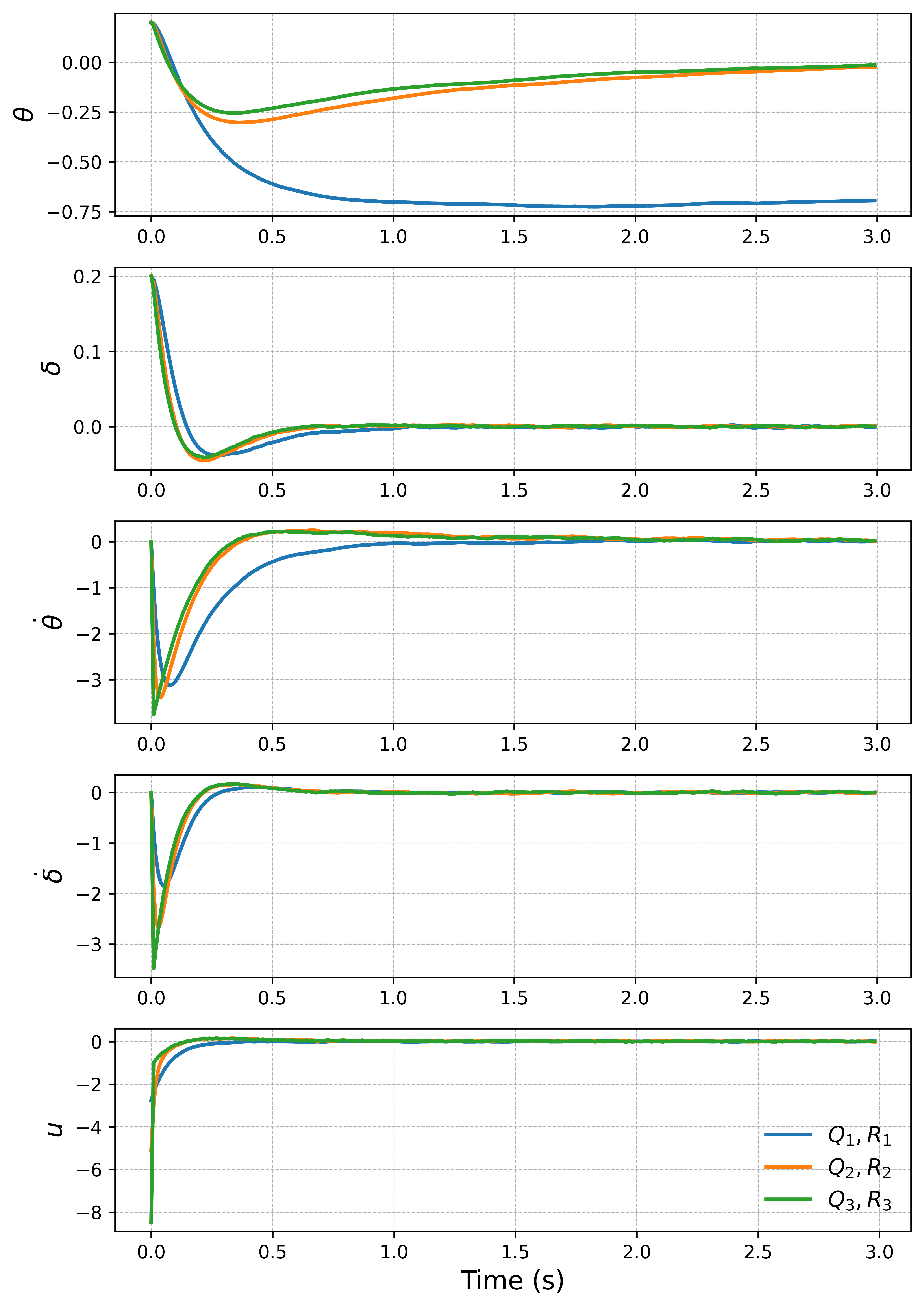}}
\caption{The evolution of the state variables (top four plots) and the control input (the fifth plot) using LQR approach with various $Q$ and $R$.}
\label{fig:lqr}
\end{figure}

The control input \( u  \) (V) is the voltage applied to the motor. Here, we considered linearized dynamics around its upright equilibrium point.
The control strategy aims to stabilize the pendulum in its upright equilibrium while ensuring smooth control effort. Both the natural gradient control method and LQR are employed to regulate the trajectory of the state vector. The discrete-time state-space matrices used in the simulation are given by
\begin{align*}
A &=
\begin{bmatrix}
1.0 &0.0073 & 0.0094 & 0 \\
0 & 1.012 & -0.0006 & 0.01\\
0 & 1.43 & 0.89 & 0.0026 \\
0 & 2.55 & -0.12 & 1.00
\end{bmatrix}, 
B &=
\begin{bmatrix}
0.0024 \\
0.0024 \\
0.48  \\
0.47
\end{bmatrix}.   
\end{align*}
We compared the natural gradient control method~\eqref{eq:theorem1} for different values of the step size parameter $\alpha$ with the standard LQR for various choices of weighting matrices $Q$ and $R$.
For the natural gradient control method, we choose $\lambda=0.99$, and the step size values
$\alpha \in \{0.01, 0.018, 0.025\}.
$
For the LQR controller, we considered different state and control weighting matrices
$Q_1 = 0.001 I$,  $Q_2 = I$, $Q_3 = 5000 I$,$R_1 = 100$, $R_2=1$, and $R_3 =  0.001.$
The comparison between the natural GD approach and LQR control highlights key differences in trajectory shaping and optimal control. In the natural GD method, state evolution is influenced by the step size parameter 
\( \alpha \), with larger values leading to faster convergence but potentially introducing oscillations. This approach provides direct control over trajectory shaping, allowing for intuitive adjustments to the system’s behavior. 

Figures \ref{fig:ngd} and \ref{fig:lqr} show the results of the proposed and LQR approaches, respectively. LQR control optimally balances state regulation and control effort through the careful selection of weighting matrices \( Q \) and \( R \), ensuring smooth and stable convergence. However, a fundamental limitation of LQR is that the trajectory behavior is not explicitly predictable based on \( Q \) and \( R \) adjustments. While these matrices influence the control law, their effect on the system’s transient response is often indirect and requires trial-and-error tuning. Consequently, achieving a specific trajectory shape in LQR can be challenging, as the resulting state evolution emerges from the Riccati equation rather than being explicitly controlled.
On the other hand, the GD-based approach offers a geometric perspective on control, where trajectory updates follow a well-defined GD process. This enables a more transparent and predictable way to shape trajectories directly. If precise trajectory shaping is the priority, the natural GD approach provides greater interpretability. In contrast, LQR remains a reliable choice for energy-efficient and well-balanced state regulation but at the cost of requiring heuristic tuning for transient behavior control.


\section{Conclusion and Future Directions}\label{sec:conclusion}
This paper introduced a novel closed-loop control framework based on natural GD, using the closed-loop covariance matrix as a preconditioner. By directly shaping system trajectories through a gradient-descent-like update, the proposed approach eliminates the need for indirect cost-function tuning, offering improved interpretability. Theoretical analysis established stability conditions and step-size constraints.

Future research directions include extending the framework to nonlinear and time-varying systems, incorporating state-dependent covariance adaptation, and exploring its integration with reinforcement learning and model predictive control. Additionally, investigating robustness properties under model uncertainties and external disturbances could further enhance the practical applicability of this approach in real-world scenarios.

\bibliography{ifacconf}             
                                                   
\end{document}